\documentclass[11pt, a4paper]{article}
\usepackage{fullpage}

\usepackage{belov_paper}
\usepackage{algorithm,algpseudocode}

\renewcommand{\eps}{\epsilon}

\newcommand{\MakeAlias}[2]{\newenvironment{#1}{\begin{#2}}{\end{#2}}}

\MakeAlias{definition}{defn}
\MakeAlias{theorem}{thm}
\MakeAlias{lemma}{lem}
\MakeAlias{claim}{clm}
\MakeAlias{proposition}{prp}
\MakeAlias{corollary}{cor}
\MakeAlias{conjecture}{conj}
\MakeAlias{remark}{rem}

\newcommand{\Hybrid}{\mathrm{Hybrid}}
\newcommand{\R}{\mathbb{R}}

\title{A Polynomial Lower Bound for Testing Monotonicity}
\author{
  Aleksandrs Belovs \\
  CWI \\
  Amsterdam, the Netherlands \\
  \texttt{stiboh@gmail.com} 
    \and 
  Eric Blais \\
  University of Waterloo \\
  Waterloo, Canada \\
  \texttt{eric.blais@uwaterloo.ca}
}

\begin{document}
\maketitle

\begin{abstract}
	We show that every algorithm 
	for testing $n$-variate Boolean functions for monotonicity
	must have query complexity $\tilde{\Omega}(n^{1/4})$.
	All previous lower bounds for this problem
	were designed for non-adaptive algorithms and, as a result,
	the best previous lower bound for general (possibly adaptive) monotonicity
	testers was only $\Omega(\log n)$.
	Combined with the query complexity of the 
	non-adaptive monotonicity tester of Khot, Minzer, and Safra (FOCS 2015),
	our lower bound
	shows that adaptivity can result in at most a quadratic reduction in the 
	query complexity for testing monotonicity.

	By contrast, we show that there is an exponential gap
	between the query complexity of adaptive and non-adaptive algorithms for testing
	regular linear threshold functions (LTFs) for monotonicity.
	 Chen, De, Servedio, and Tan (STOC 2015)
	recently showed that non-adaptive algorithms require almost 
	$\Omega(n^{1/2})$ queries for this task. We
	introduce a new adaptive monotonicity testing algorithm which has
	query complexity $O(\log n)$ when the input is a regular LTF.
\end{abstract}

\setcounter{page}{0}
\thispagestyle{empty}
\newpage

\section{Introduction}

The Boolean function $f\colon\cube\to\bool$ is {\em monotone} iff $f(x)\le f(y)$ for all $x\preceq y$, where $\preceq$ is the bitwise partial order on the Boolean hypercube $\cube$ (i.e., $x\preceq y$ iff $x_i\le y_i$ for all $i\in[n]$).
Conversely, the function $f$ is {\em $\eps$-far from monotone} for some $\epsilon > 0$ if for every monotone function $g\colon\cube\to\bool$, 
there are at least $\epsilon 2^n$ points $x \in \cube$ such that
$f(x) \neq g(x)$.
An {\em $\eps$-tester for monotonicity} is a bounded-error randomized algorithm that distinguishes monotone functions from those that are $\eps$-far from monotone.
The tester has oracle access to the function $f$.
It is \emph{non-adaptive} if its queries do not depend on the oracle's responses to the previous queries; otherwise, it is \emph{adaptive}.

The study of the monotonicity testing problem was initiated in 1998 by Goldreich, Goldwasser, Lehman, and Ron~\cite{goldreich:monotonicityTesting}, who introduced
the natural \emph{edge tester} for monotonicity. This tester selects
edges $x \preceq y$ of the hypercube $\cube$ uniformly at random 
and verifies that $f(x) \le f(y)$ on each of these edges. In the journal version of the paper~\cite{goldreich:monotonicityTesting}, they showed that this tester has query complexity
$O(n/\epsilon)$, proved that their analysis of the algorithm was tight, and asked: are there any other $\epsilon$-testers for monotonicity with significantly smaller query complexity?

\subsection{Previous work on monotonicity testing}

In 2002, Fischer \etal~\cite{fischer:monotonicitytesting} showed that 
every non-adaptive tester for monotonicity has query complexity $\Omega(\log n)$.%
\footnote{Throughout the paper, we assume that $\eps=\Theta(1)$ in the lower bound settings.}
This immediately implies an $\Omega(\log\log n)$ lower bound for the more general class of adaptive testers for monotonicity. 
Since then, stronger lower bounds were established for more restricted classes of algorithms, like 1-sided non-adaptive algorithms~\cite{fischer:monotonicitytesting} and even more limited \emph{pair testers}~\cite{briet:monotonicityAndRouting}---algorithms that select pairs 
$x \preceq y$ of inputs from some distribution over the comparable pairs of 
inputs in the hypercube and check that $f(x) \le f(y)$ on each selected pair. 
Algorithms and strong lower bounds were also introduced for the related problem 
of testing monotonicity of functions 
with non-Boolean ranges~\cite{dodis:improvedTestingMonotonicity, blais:testingLowerViaCommunication, chakrabarty:OptimalLowerBoundMonotonicity}.
However, there was no further progress on Goldreich \etal's original question 
for more than a decade, until a recent outburst of activity.

In 2013, Chakrabarty and Seshadhri~\cite{chakrabarty:sublinearMonotonicity}
showed that there are indeed testers for monotonicity with query complexity
asymptotically smaller than that of the edge tester. 
They introduced a pair tester with query complexity $\tilde{O}(n^{7/8}\eps^{-3/2})$.
Chen, Servedio and Tan~\cite{chen:newAlgorithmsLowerBoundsMonotonicity} further developed these ideas to obtain a pair tester with query complexity
$\tO(n^{5/6}\eps^{-4})$.
Khot, Minzer, and Safra~\cite{khot:monotonicityTesting} 
showed that a directed version of Talagrand's isoperimetric inequality yields a 
pair tester with query complexity $\tO(\sqrt{n}/\eps^2)$.
The authors~\cite{belovs:monotonicityQuantum} used this inequality to develop a \emph{quantum} tester for 
monotonicity with query complexity $\tO(n^{1/4}\eps^{-1/2})$. 

On the lower bound side, Chen, Servedio and Tan~\cite{chen:newAlgorithmsLowerBoundsMonotonicity} established a lower bound of $\tOmega(n^{1/5})$ queries for all non-adaptive testers for monotonicity.
This lower bound was later improved to almost $\Omega(\sqrt{n})$ by Chen, De, Servedio and Tan~\cite{chen:booleanMonotonicityRequiresSqrtn}.
These recent developments essentially give a complete answer to the
question of Goldreich \etal~for non-adaptive algorithms: there exists a 
non-adaptive tester for monotonicity with query complexity that is quadratically
smaller than that of the edge tester, and this gap is best possible.

\subsection{Our results}

\mytxtcommand{chen}{Chen--Servedio--Tan}
\mytxtcommand{chenDe}{Chen--De--Servedio--Tan}

Despite all the recent progress on monotonicity, our understanding of 
the query complexity of \emph{adaptive} testers for monotonicity remains 
far from complete.
The best lower bound for the problem
is $\Omega(\log n)$, which follows
directly from the non-adaptive lower bound of 
Chen \etal~\cite{chen:newAlgorithmsLowerBoundsMonotonicity}.
This lower bound leaves open the possibility that there exist testers for
monotonicity with query complexity that is \emph{exponentially} smaller
than that of the edge tester or of any other non-adaptive
tester for monotonicity.
Our main result eliminates this possibility.

\begin{thm}
\label{thm:main}
There exists an absolute constant $\eps>0$ such that any (adaptive) randomized algorithm that $\eps$-tests whether an $n$-variate Boolean function $f$ is monotone makes $\tOmega(n^{1/4})$ queries to $f$.
\end{thm}

Theorem~\ref{thm:main} shows that the query complexity of any tester for monotonicity 
(adaptive or not) is at most a quartic factor better than that of the edge tester, 
and that adaptivity can result in at most a quadratic reduction in the query 
complexity for the monotonicity testing problem.

The proof of Theorem~\ref{thm:main} is established by considering 
random functions known as \emph{Talagrand's random DNFs}. 
These monotone functions have previously appeared in 
many different contexts---including DNF approximation~\cite{odonnell:ApproxByDNF}, 
hardness amplification~\cite{bogdanov:HardnessAmplification}, and
learning theory~\cite{lee:LearningTalagrandDNF}---and are of particular
interest because of their extremal noise sensitivity properties~\cite{mossel:noiseSensitivityMonotone}.
We use the same noise sensitivity properties to show that Talagrand's random DNF with $\sqrt{n}$ random input variables negated is $\Omega(1)$-far from monotone with high probability, and that 
a randomized algorithm 
with small query complexity cannot reliably distinguish original Talagrand's random 
DNFs from this modified version. 

Our approach represents a notable departure from previous lower bounds for the 
monotonicity testing problem, in that all the previous lower bounds~\cite{fischer:monotonicitytesting, chen:newAlgorithmsLowerBoundsMonotonicity, chen:booleanMonotonicityRequiresSqrtn} were obtained by considering linear
threshold functions (LTFs)---Boolean functions of the form
$f(x) = \sgn(\sum_{i\in [n]} w_i x_i - \theta)$ with appropriate \emph{weight} $w_1,\ldots,w_n \in \R$ and \emph{threshold} $\theta \in \R$ parameters.
In fact, the previous lower bounds for monotonicity testing were obtained
by considering a special class of LTFs known as \emph{regular LTFs}.
An LTF is \emph{$\tau$-regular} when the magnitude of each weight
$w_i$ is bounded by $|w_i| \le \tau \cdot \sqrt{\sum_{j \in [n]} w_j^{\,2}}$.
Regular LTFs have been studied in the context of approximating~\cite{diakonikolas:ImprovedApproximationLTF}, learning~\cite{odonnell:ChowParameters}, 
and testing~\cite{matulef:testingHalfspaces} LTFs;
the lower bounds in~\cite{chen:newAlgorithmsLowerBoundsMonotonicity, chen:booleanMonotonicityRequiresSqrtn} are obtained by showing that
non-adaptive algorithm with small query complexity cannot reliably distinguish
$O(\frac1{\sqrt{n}})$-regular LTFs that are monotone from those that are
far from monotone.

Chen, De, Servedio, and Tan~\cite{chen:booleanMonotonicityRequiresSqrtn} asked
if their approach could be generalized to obtain polynomial lower bounds on the 
query complexity of adaptive testers for monotonicity. We answer this question in
the negative, by showing that there does exist an adaptive algorithm
with logarithmic query complexity that can $\epsilon$-test monotonicity 
when its input is promised to be a regular LTF.

\begin{theorem}
\label{thm:regular-LTFs}
	Fix $\epsilon > 0$ and $\tau > 0$.
	There is an adaptive algorithm $\cA$ with query complexity%
	\footnote{\label{foot:middle} In fact, we can restrict $\cA$ to only query the value of the function
	on inputs from the middle layers of the hypercube, 
	so it also $\epsilon$-tests \emph{truncated} regular LTFs for monotonicity. 
	See \rf(defn:truncate) for the definition of truncation, and \rf(sec:algBalance) for more details.}
	$O_{\eps,\tau}(1) + \log n$ that, given oracle access to the $n$-variate Boolean function $f$,
	\begin{enumerate}
		\vspace{-5pt}
	 	\item Always accepts when $f$ is a monotone $\frac\tau{\sqrt{n}}$-regular LTF, and
	 	\vspace{-10pt}
	 	\item Rejects with probability at least $\frac12$ when $f$ is a $\frac\tau{\sqrt{n}}$-regular LTF that is $\epsilon$-far from monotone.
	\end{enumerate} 
\end{theorem}

Combined with the lower bound of Chen et al.~\cite{chen:booleanMonotonicityRequiresSqrtn}, Theorem~\ref{thm:regular-LTFs} shows that
there are natural classes of functions for which adaptivity can reduce the query complexity of monotonicity testers by an exponential amount. By the standard reduction between adaptive and non-adaptive algorithms, this is best possible.

The proof of Theorem~\ref{thm:regular-LTFs} is obtained by introducing a new
adaptive tester for monotonicity. The algorithm is quite natural: it selects pairs
of inputs $x,y \in \cube$ independently at random until it finds a pair for 
which $f(x) \neq f(y)$, then it performs a random binary search between $x$ and $y$ to identify an edge $(z,z')$ of the hypercube on which $f(z) \neq f(z')$.
The algorithm accepts if and only if $f$ is monotone on this edge.
To the best of our knowledge, this algorithm is the first randomized adaptive tester for monotonicity on the hypercube to be analyzed. (See, for example, the discussions in~\cite[\S 1.5]{khot:monotonicityTesting} and in~\cite{canonne:OpenProblem}.)



\paragraph{Organization.}
We discuss the proofs of Theorems~\ref{thm:main} and~\ref{thm:regular-LTFs} at a 
high-level in Section~\ref{sec:intuition}, after introducing preliminary facts
and terminology. The complete proofs follow
in Sections~\ref{sec:lower} and~\ref{sec:upper}, respectively.

\section{Preliminaries}
\label{sec:preliminaries}



\subsection{Probability Theory}

We use standard concentration inequalities.

\begin{lemma}[Hoeffding's inequality]
 Let $w \in \R^n$ be any real-valued vector. 
 Then for any $t > 0$,
 when $X_1,\ldots,X_n$ are independent random variables taking the values
 $+1$ and $-1$ with probability $\frac12$ each,
 $$
 \Pr\skC[ \absB|\sum_{i \in [n]} w_i X_i| > t] \le 2\ee^{-\frac{t^2}{2\|w\|_2^2}},
 $$
where $\norm|w|_2 = \sqrt{\sum_{i\in[n]} w_i^{\,2}}$ is the $\ell_2$-norm.
\end{lemma}

\begin{lemma}[Bernstein's inequality]
Consider a set of $n$ independent random variables $X_1,\dots,X_n$, where $-1\le X_i\le 1$ for all $i$.
Let $X = \sum_{i\in[n]} X_i$.
Then, for all $0 < t<\Var[X]$, we have
\[
\Pr\skB[{\absA|X - \bE[X]| > t }] \le 2 \ee^{-\frac{t^2}{4\Var[X]}}.
\]
\end{lemma}

We also use an anti-concentration inequality that follows directly from
the Berry--Ess{\'e}en theorem. (See, e.g.,~\cite{odonnell:ChowParameters}.)

\begin{lemma}[Berry--Ess{\'e}en corollary]
\label{lem:berry-esseen-cor}
Fix $\tau > 0$. Let $w \in \R^n$ be any real-valued vector that satisfies
$\max_j |w_j| \le \tau \|w\|_2$. Then for any $a < b \in \R$,
when $X_1,\ldots,X_n$ are independent random variables taking the values
$+1$ and $-1$ with probability $\frac12$ each,
$$
\Pr\skC[ a \le \sum_{i \in [n]} w_i X_i \le b] \le \frac{b-a}{\|w\|_2} + 2 \tau.
$$
\end{lemma}

\subsection{Property testing lower bounds}
\label{sec:testingLBs}

\mycommand{yes}{\mathsf{Yes}}
\mycommand{no}{\mathsf{No}}

Theorem~\ref{thm:main} is established via a standard lemma
concerning the general setting where $\cP$ and $\cN$ are two disjoint 
families of $n$-variate Boolean functions, 
an algorithm is given oracle access to a function $f\in\cP\cup\cN$, 
and its task is to determine whether $f\in\cP$ or $f\in\cN$.
The following lemma is essentially folklore---see, e.g.,~\cite{fischer:art} for usage in property testing and~\cite{razborov:disjointness} for a related lemma. 
We include a short proof for completeness.

\begin{lem}
\label{lem:adaptiveLB}
Let $\yes$ and $\no$ be probability distributions on $n$-variate Boolean functions satisfying
\[
\Pr_{f\sim\yes}\skA[f\in\cP] = 1
\qquad\text{and}\qquad
\Pr_{g\sim\no}\skA[g\in\cN] = \Omega(1).
\]
If $q$ is a positive integer such that for any sequences $x_1,\dots,x_q\in\cube$ and $b_1,\dots,b_q \in\bool$, 
\begin{equation}
\label{eqn:adaptiveLB}
\Pr_{f\sim\yes} \skB[\forall i\colon f(x_i) = b_i]
\le 
(1+o(1)) \Pr_{g\sim\no} \skB[\forall i\colon g(x_i) = b_i] + o(2^{-q}),
\end{equation}
then any randomized algorithm that decides whether $f\in \cP$ or $f\in \cN$ makes $\Omega(q)$ queries to $f$.
\end{lem}

\pfstart
Let $\cA$ be a randomized decision tree that distinguishes $\cP$ from $\cN$.
Denote $p = \Pr_{g\sim\no}\skA[g\in\cN] = \Omega(1)$.
With a constant number of repetitions of $\cA$, we may assume that $\cA$ accepts any function $f\in\cP$ with probability at least $1 - p/2$, and accepts each $g\in\cN$ with probability at most $1/3$.
Then,
\[
\Pr_{f\sim\yes} \skA[\text{$\cA$ accepts on $f$}] \ge 1-\frac p2
\qquad
\text{and}
\qquad
\Pr_{g\sim\no} \skA[\text{$\cA$ accepts on $g$}] \le (1-p)+\frac p3 = 1-\frac{2p}3,
\]
Assume towards a contradiction that $\cA$ makes at most $q$ queries.  As $\cA$ is a probability distribution on deterministic decision trees, there exists a decision tree $\cD$ of depth at most $q$ such that
\begin{equation}
\label{eqn:LB1}
\Pr_{f\sim\yes} \skA[\text{$\cD$ accepts on $f$}] - \Pr_{g\sim\no} \skA[\text{$\cD$ accepts on $g$}] \ge \frac p6.
\end{equation}

Without loss of generality, we may assume that every leaf of $\cD$ is at depth exactly $q$.
Let $L$ denote the set of leaves of $\cD$.  Each leaf $\ell\in L$ is characterized by two sequences $x_1,\dots,x_q\in\cube$ and $b_1,\dots,b_q \in\bool$ such that $\cD$ ends its work in $\ell$ on $f$ iff $f(x_i)=b_i$ for all $i$.  Let $L_1\subseteq L$ be the set of leaves on which $\cD$ accepts.  Then, by~\rf(eqn:adaptiveLB),
\begin{align*}
\Pr_{f\sim\yes} \skA[\text{$\cD$ accepts on $f$}]
&=\sum_{\ell\in L_1} \Pr_{f\sim\yes} \skA[\text{$\cD$ terminates in $\ell$ on $f$}]\\
&\le \sA[1+o(1)] \sum_{\ell\in L_1} \Pr_{g\sim\no} \skA[\text{$\cD$ terminates in $\ell$ on $g$}] + o\sA[|L_1|2^{-q}]\\
&= \Pr_{g\sim\no} \skA[\text{$\cD$ accepts on $g$}] + o(1),
\end{align*}
contradicting~\rf(eqn:LB1).
Hence, $\cA$ makes $\Omega(q)$ queries.
\pfend

\mycommand{truncate}{\mathrm{Truncate}}
The following operation is often useful in lower bounds on monotonicity on the hypercube.
It essentially reduces monotonicity testing on the whole hypercube to monotonicity testing on its middle layers.

\begin{definition}
\label{defn:truncate}
	For $\delta > 0$, the \emph{$\delta$-truncation} of the function
 	$f\colon\cube\to\bool$ is the function $\truncate_\delta(f)$ defined by
\[
x\mapsto
\begin{cases}
0,& \text{if $\abs|x|< \frac n2 - \delta\sqrt n$;}\\
f(x),& \text{if $\frac n2 - \delta\sqrt n\le |x|\le \frac n2 + \delta\sqrt n$;}\\
1,& \text{if $|x|> \frac n2 + \delta\sqrt n$;}
\end{cases}
\]
\end{definition}

When $f$ is monotone, then $\truncate_\delta(f)$ is also monotone.
Furthermore, for every $\eps>0$, there exists $\delta>0$ such that $\truncate_\delta(f)$ is $\frac\eps2$-far from monotone whenever $f$ is $\eps$-far from monotone.
Note that it only makes sense to query $\truncate_\delta(f)$ on the inputs $x\in\cube$ satisfying $|x| = \frac n2\pm O(\sqrt n)$, since otherwise the response is known in advance.  We call such inputs \emph{nearly balanced}.

\subsection{Linear Threshold Functions}

In studying linear threshold functions, it is more convenient to assume that the function is of the form $f\colon\sbool^n\to\sbool$.

\begin{definition}[LTF]
	The function $f\colon \{-1,1\}^n \to \{-1,1\}$ is a \emph{linear threshold function}
	(alternatively: \emph{LTF}, or \emph{halfspace}) with associated \emph{weights}
	$w_1,\ldots,w_n \in \R$ and \emph{threshold} $\theta$ if it satisfies
	$$
	f(x) = \sgn \sC[\sum_{i=1}^n w_i x_i - \theta ]
	$$
	for every $x \in \{-1,1\}^n$ where $\sgn$ is the sign function defined by $\sgn(x)={\bf 1}[x\ge 0]$.
\end{definition}


\begin{definition}[Regular LTFs]
	The LTF $f\colon \{-1,1\}^n \to \{-1,1\}^n$ is \emph{$\tau$-regular} 
	if it can be represented with a set of weights $w_1,\ldots,w_n$ 
	that satisfy
	$
	\max_{i \in [n]} |w_i| \le \tau \cdot \sqrt{\sum\nolimits_{i=1}^n w_i^{\,2}}.
	$
\end{definition}

\subsection{Noise Sensitivity and Talagrand's Random DNFs}
\label{sec:talagrand}
\mycommand{NS}{\mathrm{NS}}
\mycommand{Tal}{\mathsf{Tal}}
Let $\cB(n, \delta)$ be the probability distribution on the subsets of~$[n]$, in which each element is included in the subset independently with probability~$\delta$.

\begin{definition}
The \emph{noise sensitivity} of a function $f\colon\cube\to\bool$ at noise rate~$\delta$ is
\[
\NS_\delta(f) = \Pr_{x\sim\cube,\; S\sim\cB(n,\delta)} \sk[f(x)\ne f(x^S)],
\]
where $x^S$ denotes the input string $x$ with the variables in $S$ flipped.
\end{definition}

Talagrand's random DNF
 on $n$ variables~\cite{talagrand:randomCNF} is a disjunction of $2^{\sqrt{n}}$ independent random clauses of size $\sqrt{n}$.%
\footnote{
Talagrand's original definition was for random \emph{CNF}s.
However, DNFs are more convenient than CNFs for our intended applications, and all the results about CNFs easily carry over to the DNF case by duality.
}
More precisely, let $\cC$ be the uniform probability distribution on functions $C\colon [\sqrt{n}]\to[n]$.
We identify each $C$ in $\cC$ with the Boolean function $f_C\colon \cube\to\bool$ given by 
$
f_C(x) = \bigwedge_{a\in [\sqrt{n}]} x_{C(a)}.
$
\emph{Talagrand's random DNF} $f$ is then defined as 
\[
f(x) = \bigvee_{j\in [2^{\sqrt{n}}]} f_{C_j}(x),
\]
where each clause $C_j$ is independently sampled from $\cC$.
Let us denote the distribution of $n$-variate Talagrand's random DNF by $\Tal$.

One of the particularly useful characteristics of Talagrand's random DNF is that it  is one of the most noise-sensitive monotone functions, as shown by the following result.%
\footnote{
Mossel and O'Donnell only postulate the existence of one such function $f$.
However, \rf(thm:talagrand) easily follows from the equation before the Proof of Theorem 3 in Section 4 of~\cite{mossel:noiseSensitivityMonotone}.
}

\begin{thm}[Mossel-O'Donnell~\cite{mossel:noiseSensitivityMonotone}]
\label{thm:talagrand}
Talagrand's random DNF $f$ satisfies $\NS_{1/\sqrt{n}}(f) = \Omega(1)$ with probability $\Omega(1)$.
\end{thm}

\section{High-level overview and intuition}
\label{sec:intuition}


\subsection{The Bisection Algorithm and Regular LTFs}

The intuition behind the proofs of Theorems~\ref{thm:main} and~\ref{thm:regular-LTFs}
is best described by first examining the previous non-adaptive query complexity
lower bounds of Chen \etal~\cite{chen:newAlgorithmsLowerBoundsMonotonicity, chen:booleanMonotonicityRequiresSqrtn}. In these lower bounds, two
distributions $\cD_\yes$ and $\cD_\no$ over a finite set of weights are defined
under the two constraints that
\begin{enumerate}
	\item Every weight in the support of $\cD_\yes$ is non-negative, and
	\vspace{-8pt}
	\item A weight $w \sim \cD_\no$ is negative with constant probability.
\end{enumerate}
Two distributions $\yes$ and $\no$ over $n$-variate LTFs are defined by 
drawing weights $w_1,\ldots,w_n$ independently at random from the distributions
$\cD_\yes$ and $\cD_\no$, respectively, and then by letting
\[
f(x_1,\dots,x_n) = \sgn(w_1x_1 + \cdots + w_n x_n).
\]
Since $\cD_\yes$ and $\cD_\no$ are over finite domains (of size independent of $n$),
the resulting function $f$ is always an $O(\frac1{\sqrt n})$-regular LTF~\cite[Claim B.2]{chen:booleanMonotonicityRequiresSqrtn}.
Furthermore, the functions drawn from $\cD_\yes$ are always monotone, and
the functions drawn from $\cD_\no$ are $\Omega(1)$-far from monotone with
large probability~\cite[Theorem B.9]{chen:booleanMonotonicityRequiresSqrtn}. Thus, we have the following consequence:

\begin{thm}[\chenDe~\cite{chen:booleanMonotonicityRequiresSqrtn}]
For each $\delta>0$, there exist $\eps,\tau=\Theta(1)$ such that $\Omega(n^{1/2-\delta})$ non-adaptive nearly balanced queries are required to $\eps$-test $\frac\tau{\sqrt n}$-regular LTFs for monotonicity.
\end{thm}

Regular LTFs are used in the proofs of~\cite{chen:newAlgorithmsLowerBoundsMonotonicity, chen:booleanMonotonicityRequiresSqrtn} because
with suitable weight distributions $\cD_\yes$ and $\cD_\no$, 
appropriate central limit theorems can be used to bound the query complexity of
non-adaptive algorithms.  
Regular LTFs, however, also have one other 
notable characteristic: when a $O(\frac1{\sqrt{n}})$-regular LTF is far from monotone, then a \emph{constant} fraction of the edges
$x \preceq y$ of the hypercube on which $f(x) \neq f(y)$ are edges where
$f(x) > f(y)$ and are thus witnesses to the non-monotonicity of $f$.

This observation suggests a natural approach for testing monotonicity
of regular LTFs: draw an edge $x \preceq y$ uniformly at random from the set
of edges where $f(x) \neq f(y)$, and test whether $f$ is monotone on this edge.
While we unfortunately do not know of any query-efficient algorithm for drawing edges 
from this distribution, we do know of one way to at least guarantee that we
return \emph{some} edge $x \preceq y$ on which $f(x) \neq f(y)$ using a
logarithmic number of queries when $f$ is not too biased.
A simple way to do this is described in the bisection algorithm below. In this algorithm, for $x,y\in\cube$, $\Hybrid(x,y)$ denotes the set of inputs $z\in\cube$ that satisfy $z_i = x_i$ for every index $i \in [n]$ where $x_i=y_i$.

\begin{algorithm}
	\caption{Bisection algorithm}
	\label{alg:bisection}
	\begin{algorithmic}[1]
		\State Draw $x,y \in \cube$ uniformly and independently at random 
			   until $f(x) = 0$ and $f(y) = 1$.
	    \State If $O(1/\epsilon)$ pairs are drawn without satisfying the condition,
		       {\bf accept}.
		\While{ $|\Hybrid(x,y)| > 2$ }
			\State Draw $z \in \Hybrid(x,y)$ uniformly at random.
			\State If $f(z) = 0$, update $x \gets z$.
			\State Otherwise if $f(z) = 1$, update $y \gets z$.
		\EndWhile
		\State If $x \preceq y$, {\bf accept}; otherwise {\bf reject}.
	\end{algorithmic}
\end{algorithm}

The proof of Theorem~\ref{thm:regular-LTFs} is completed by showing that 
a slight variant of this algorithm does indeed identify a non-monotone edge
with constant probability
when the input function is a regular LTF that is far from monotone.
Specifically, we consider the random process on subsets of $[n]$ defined
by the bisection algorithm and show that with constant probability, after
$\log n - \Theta(1)$ iterations of the while loop, the set
$\{i \in [n] : x_i \neq y_i\}$ has cardinality $O(1)$ and
contains some coordinates with negative weights.
The details are in \rf(sec:upper).


\subsection{Noise Sensitivity and Polynomial Lower Bound}
\label{sec:intuitionLower}

Theorem~\ref{thm:regular-LTFs} shows that we need other functions than regular
LTFs to prove a polynomial lower bound for adaptive monotonicity testing.
To find such functions, we can start by identifying functions that are
far from monotone but for which the bisection algorithm rejects only with small
probability.

On a function $f\colon\cube\to\bool$, the bisection algorithm ends its work in an edge $xy$ of the hypercube, where $f(x)\ne f(y)$.
Let us say in this case that the algorithm ends its work in variable $i$, where $i$ is the only variable where $x$ and $y$ differ.
Thus, on each $f$, the bisection algorithm defines the corresponding \emph{output probability distribution} on the variables in $[n]$.
Our first observation is that negating some input variables of a function does
not affect the output probability distribution of the bisection algorithm.

\begin{prp}
\label{prp:outputDistrib}
For each $f\colon\cube\to\bool$ and $S\subseteq[n]$, the output probability distributions on $[n]$ defined by the bisection algorithm on the functions $f$ and $g(x) = f(x^S)$ are identical.
\end{prp}

\pfstart
Let 
\[
(x_1,y_1), (x_2,y_2),\dots, (x_t, y_t)
\]
be a transcript of the bisection algorithm on the function $f$.  That is, $(x_i,y_i)$ is the value of $x$ and $y$ before the $i$th iteration of the loop in \rf(alg:bisection).
Then,
\[
(x_1^S,y_1^S), (x_2^S,y_2^S),\dots, (x_t^S, y_t^S)
\]
is an equiprobable transcript of the bisection algorithm on the function $g$, which ends its work in the same variable.
\pfend

Our next observation is that if we have a monotone function with large noise sensitivity, then negating a (small) random subset of the variables yields
a function that that is far from monotone with high probability.

\begin{lem}
\label{lem:noise->far}
Let $f\colon\cube\to\bool$ be a monotone function and $0<\delta<1$ be a real number.
Assume $\NS_{\delta}(f) = \Omega(1)$.
Then, with probability $\Omega(1)$ over the choice of $S\sim\cB(n,\delta)$, the function $g(x) = f(x^S)$ is $\Omega(1)$-far from being monotone.
\end{lem}

\pfstart
By the definition of noise sensitivity, 
\[
\Pr_{x\sim\cube,\; S\sim\cB(n,\delta)} \sk[f(x) \ne f(x^S)] = \Omega(1).
\]
By Markov's inequality, with probability $\Omega(1)$ over the choice of $S\sim\cB(n,\delta)$, we have
\begin{equation}
\label{eqn:noise1}
\Pr_{x\sim\cube} \sk[f(x)\ne f(x^S)] = \Omega(1).
\end{equation}
Let $g(x) = f(x^S)$ be defined for such an $S$, and let $D(g)$ denote the number of inputs on which we have to modify the value of $g$ in order to make it monotone.  We aim to estimate $D(g)$.\

Write $x=(y,z)$ with $y\in\bool^{[n]\setminus S}$ and $z\in\bool^S$.
For each $y$, consider the function $g_y(z) = g(y,z)$.
We have $D(g) \ge \sum_y D(g_y)$.
Next, each $g_y$ is anti-monotone.
This implies $D(g_y)\ge \min\{g_y^{-1}(0), g_y^{-1}(1)\}$.
We can lower bound the latter quantity by the number of pairs $\{z, z^S\}$ satisfying $g_y(z)\ne g_y(z^S)$.
Summing over all $y$, we get that $D(g)$ is at least the number of pairs $\{x,x^S\}$ satisfying $f(x) \ne f(x^S)$.
By~\rf(eqn:noise1), $g$ is $\Omega(1)$-far from being monotone.
\pfend

These observations, along with \rf(thm:talagrand), show that there are indeed
functions that are far from monotone but are rejected by the bisection algorithm
with only a small probability.

\begin{prp}
There exists a boolean function $g\colon\cube\to\bool$ that is $\Omega(1)$-far from being monotone, but such that the bisection algorithm rejects $g$ with probability only $O(1/\sqrt n)$.
\end{prp}

\pfstart
Let $f$ be a monotone functions satisfying $\NS_{1/\sqrt{n}}(f) = \Omega(1)$.
By \rf(thm:talagrand), a Talagrand random DNF satisfies this condition with
probability $\Omega(1)$.
Let $p_i$ be the output probability of variable $i\in [n]$ defined by the bisection algorithm on $f$.
Let $S\sim \cB(n,1/\sqrt n)$.  Then,
\[
\bE_{S}\skC[\sum_{i\in S} p_i] = \frac1{\sqrt n}.
\]
By Markov's inequality, and using \rf(lem:noise->far), there exists $S$ such that 
the function $g(x) = f(x^S)$ is $\Omega(1)$-far from being monotone, \emph{and}
$\sum_{i\in S} p_i = O(1/\sqrt n)$.
By \rf(prp:outputDistrib), the latter sum is exactly equal to the rejection probability of the bisection algorithm on the function $g$.
\pfend

This result shows that there are functions obtained by negating some variables in a Talagrand random DNF that are $\Omega(1)$-far from monotone, but such that the bisection algorithm requires $\Omega(\sqrt n)$ queries to detect that it is non-monotone.
The proof of \rf(thm:main) uses a very different approach---after all, there is no 
direct analogue of \rf(prp:outputDistrib) that can hold for all adaptive 
algorithms---but the underlying ideas are the same. 
We show that for \emph{any} set of $q \ll n^{1/4}$ queries, the distribution of the values returned by the monotone and the non-monotone Talagrand random DNFs are very similar.
After that, we can apply \rf(lem:adaptiveLB) to complete the proof.
The high-level intuition is as follows.  Consider two queries $x,y\in\cube$.
On the one hand, if $x$ and $y$ are far from each other, then, due to noise sensitivity, the values of $f(x)$ and $f(y)$ are essentially independent.  Hence, adaptivity does not help here.
On the other hand, if $x$ and $y$ are close, they are likely to miss the set $S$ of negated input variables.
More precisely, since there are at most $q^2\ll \sqrt{n}$ pairs of close inputs, a random set $S$ of $\sqrt n$ elements will avoid all of them with high probability.
For all the details, see \rf(sec:lower).

\section{Polynomial Lower Bound}
\label{sec:lower}

In this section, we prove \rf(thm:main).
Throughout this section we use $\cB = \cB(n,1/\sqrt n)$ to denote the probability distribution on subsets of $[n]$ where each element is included in the subset independently with probability $1/\sqrt{n}$.
Following the discussion in  \rf(sec:intuitionLower), let us define the distribution $\Tal^\pm$ of \emph{Talagrand's random non-monotone DNFs} as the following distribution on $n$-variate Boolean functions
\[
\Tal^\pm = \sfig{x\mapsto f(x^S)\midA f\sim\Tal,\; S\sim\cB}.
\]
We define two distributions for a sufficiently large constant $\delta>0$:
\[
\yes = \sfig{\truncate_\delta(f)\mid f\sim\Tal}
\qquad\text{and}\qquad
\no = \sfig{\truncate_\delta(f)\mid f\sim\Tal^\pm}.
\]
In view of \rf(lem:adaptiveLB), \rf(thm:talagrand) and \rf(lem:noise->far), it suffices to show that for all $q = O(n^{1/4}\log^{-2} n)$, nearly balanced input strings $x_1,\dots,x_q\in\cube$ and Boolean outcomes $b_1,\dots,b_q\in\bool$, we have
\begin{equation}
\label{eqn:patterns}
\Pr_{f\sim\Tal} \skB[\forall i\colon f(x_i) = b_i]
\le 
(1+o(1)) \Pr_{g\sim\Tal^\pm} \skB[\forall i\colon g(x_i) = b_i] + o(2^{-q}).
\end{equation}

\subsection{Proof of Theorem~\ref{thm:main}}
\label{sec:patternsproof}

Let us denote by $X = \{x_1,\dots,x_q\}$ the set of input strings.
All of them are nearly balanced.  In this section, we often identify strings in $\cube$ with the corresponding subsets of $[n]$.

For a fixed sequence of values $b_1,\dots,b_q$, the set $X$ can be naturally divided into
\[
X_0 = \sfigA{x_i\mid i\in[q], b_i=0}
\qquad\text{and}\qquad
X_1 = \sfigA{x_i\mid i\in[q], b_i=1}.
\]
Let $J = [2^{\sqrt n}]$ be the set of indices of the clauses in Talagrand's random DNF.
Recall that a function $f$ from the $\Tal$ distribution is given by 
\[
f(x) = \bigvee_{j\in J} C_j(x),
\]
where $(C_j)\sim \cC^J$ is a sequence of random clauses.
We call the sequence $(C_j)$ \emph{compliant with respect to the shift} $T\subseteq [n]$ iff the corresponding function $f$ satisfies $\forall i\in[q]\colon f(x_i^T) = b_i$.  
\mycommand{comp}{\cM^\emptyset}
\mycommand{compS}{\cM^S}
\mycommand{compT}{\cM^T}
We denote the set of such sequences by $\compT$.
The set $\comp$ corresponds to the events on the left-hand side of~\rf(eqn:patterns), and $\compS$ for $S\sim\cB$ corresponds to the right-hand side.


For $T\subseteq[n]$, we partition the set $\compT$ in accordance to when a clause in the sequence $(C_j)$ first satisfies each particular input $x\in X_1$.
Formally, for each sequence $\tau = (\tau_x) \in J^{X_1}$, we define $\compT_\tau$ as the set of sequences $(C_j)\in\compT$ satisfying
\negbigskip
\begin{equation}
\label{eqn:tcondition}
\parbox{.9\textwidth}{
\itemstart
\item for each $x\in X_0$, we have $f_{C_j}(x^T) = 0$ for all $j\in J$;
\item for each $x\in X_1$, we have $f_{C_j}(x^T) = 0$ for all $j<\tau_x$, and $f_{C_{\tau_x}}(x^T) = 1$.
\itemend
}
\negbigskip
\end{equation}
This clearly partitions the set $\compT$ into disjoint subsets.

The conditions imposed by~\rf(eqn:tcondition) on different $C_j$ are independent, thus, we can decompose $\compT_\tau$ into the following Cartesian product:
\begin{equation}
\label{eqn:compTdecomposition}
\compT_\tau = \prod_{j\in J} \compT_{\tau,j},
\end{equation}
where $\compT_{\tau,j}$ is the projection of $\compT_\tau$ onto the $j$th component of the sequence.
Let, for $j\in J$,
\[
X_{1,j} = \{x\in X_1 \mid \tau_x = j\}
\qquad\text{and}\qquad
X_{0,j} = X_0 \cup \sfig{x\in X_1 \mid j<\tau_x}.
\]
Thus, for each $j$, we have $C\in \compT_{\tau,j}$ if and only if $f_C(x^T) = 1$ for all $x\in X_{1,j}$ and $f_C(x^T)=0$ for all $x\in X_{0,j}$.

We say that a sequence $\tau$ is \emph{good} if 
\begin{equation}
\label{eqn:X1jassumpition}
\forall j\in J,\; \forall x,y\in X_{1,j}\colon |x\cap y|\ge \frac n2 - n^{3/4},
\end{equation}
Otherwise, we call it \emph{bad}.  We treat these two cases separately.

\begin{lem}
\label{lem:badtau}
We have
\begin{equation}
\label{eqn:badtau}
\Pr_{(C_j) \sim \cC^J} \skB[\text{\rm exists bad $\tau$ such that $(C_j)\in \comp_\tau$}] = o(2^{-q}).
\end{equation}
\end{lem}

\pfstart
It is easy to see that any $(C_j)$ satisfying the condition in~\rf(eqn:badtau) also satisfies
\begin{equation}
\label{eqn:cB}
\exists j\in J,\; \exists x,y\in X\colon \sB[|x\cap y|< \frac n2 - n^{3/4}]\wedge\sB[f_{C_j}(x) = f_{C_j}(y) = 1].
\end{equation}
By the union bound, the probability that~$(C_j)$ satisfies~\rf(eqn:cB) is at most 
\[
2^{\sqrt n} q^2 \s[\frac{\frac n2 - n^{3/4}}{n}]^{\sqrt{n}} 
= q^2 \s[1-2 n^{-1/4}]^{\sqrt n} 
\le q^2 \ee^{-2n^{1/4}} = q^2 \ee^{-\Omega(\log^2 n)}\cdot 2^{-q} = o(2^{-q}).\qedhere
\]
\pfend

Let us now consider good $\tau$.
In order to prove~\rf(eqn:patterns), it suffices to show that
\begin{equation}
\label{eqn:prodComparison}
\prod_{j\in J} \absA|\comp_{\tau,j}| 
\le \sA[1+o(1)] \prod_{j\in J} \absA|\compS_{\tau,j}| 
\end{equation}
with probability $1-o(1)$ over the choice of $S\sim\cB$.
Indeed, using~\rf(eqn:compTdecomposition), we get from~\rf(eqn:prodComparison) that
\begin{align*}
\sum_{\text{$\tau$ is good}} \absA|\comp_\tau| 
&\le \sA[1+o(1)]\sum_{\text{$\tau$ is good}} \E_{S\sim\cB}\skA[|\compS_\tau|] \\
&\le \sA[1+o(1)] \E_{S\sim\cB}\skA[|\compS|]
= \sA[1+o(1)]\cdot |\cC^J| \Pr_{g\sim\Tal^\pm} \skB[\forall i\colon g(x_i) = b_i] .
\end{align*}
Hence, using also \rf(lem:badtau), we get
\begin{align*}
\Pr_{f\sim\Tal} \skB[\forall i\colon f(x_i) = b_i] &= 
\frac1{|\cC^J|}\sC[ \sum_{\text{$\tau$ is good}} \absA|\comp_\tau| + \sum_{\text{$\tau$ is bad}} \absA|\comp_\tau| ] \\
&\le  \sA[1+o(1)] \Pr_{g\sim\Tal^\pm} \skB[\forall i\colon g(x_i) = b_i] + o(2^{-q}).
\end{align*}

Let us consider~\rf(eqn:prodComparison) now.  The set of indices $J$ breaks down into two parts $J = J_1\cup J_0$, where 
\[
J_1 = \{\tau_x \mid x\in X_1\}
\qquad\text{and}\qquad
J_0 = J\setminus J_1.
\]
We prove~\rf(eqn:prodComparison) for the indices in $J_1$ and $J_0$ independently.
Indices in $J_0$ are easier to analyse because for them we have $X_{1,j}=\emptyset$.
On the other hand, we need a rather careful estimate since $|J_0|\approx 2^{\sqrt n}$.
Indices in $J_1$ are harder to analyse because for them, in general, both $X_{1,j}$ and $X_{0,j}$ are non-empty.  
But since $|J_1|\le q< n^{1/4}$, a less accurate estimate suffices.

For $J_1$, we have the following lemma, which is proven in \rf(sec:one).
\begin{lem}
\label{lem:one}
Assume $\tau$ is good. Then, for a set $S\sim\cB$, with probability $1-o(1)$, we have 
\begin{equation}
\label{eqn:one}
\prod_{j\in J_1} |\comp_{\tau,j}| 
\le \sA[1+o(1)] \prod_{j\in J_1} |\compS_{\tau,j}|.
\end{equation}
\end{lem}

For $J_0$, we have the following lemma, which we prove in \rf(sec:zero):
\begin{lem}
\label{lem:zero}
A set $S\sim\cB$ satisfies the following property with probability $1-o(1)$:
For any subset $X'\subseteq X$, we have 
\begin{equation}
\label{eqn:zero}
\Pr_{C\sim\cC} \skB[ \exists x\in X'\colon f_C(x)=1 ]
\ge \sB[1-O(n^{-1/4}\log^{3/2}n)] 
\Pr_{C\sim\cC} \skB[ \exists x\in X'\colon f_C(x^S)=1 ].
\end{equation}
\end{lem}

Assume $S$ satisfies~\rf(eqn:one) and~\rf(eqn:zero).  For $j\in J_0$, let $p_j$ and $p'_j$ denote the probability in the left- and right-hand sides of~\rf(eqn:zero), respectively, when $X' = X_{0,j}$.  
By the union bound, and since all $x\in X$ are nearly balanced,
\[
p_j' \le \sum_{x\in X_{0,j}} \Pr\skB[ f_C(x^S)=1 ] \le |X_{0,j}| \s[\frac{\frac n2 + O(\sqrt n)}n]^{\sqrt n} = O(q2^{-\sqrt n}).
\]
Thus,
\begin{equation}
\label{eqn:J01}
\frac{|\comp_{\tau,j}|}{|\compS_{\tau,j}|} = \frac{1-p_j}{1-p'_j} 
\le \frac{1-\sA[1-O(n^{-1/4}\log^{3/2}n)]p'_j}{1-p'_j} = 1 + O\s[n^{-1/4}\log^{3/2}n\cdot q2^{- \sqrt n}]
= 1 + o\sA[2^{-\sqrt n}].
\end{equation}
Hence, taking the product of~\rf(eqn:J01) over all $j\in J_0$, we get that, with probability $1-o(1)$, a set $S\sim\cB$ satisfies:
\[
\prod_{j\in J_0} \frac{|\comp_{\tau,j}|}{|\compS_{\tau,j}|}
\le \s[1+o(2^{-\sqrt n})]^{2^{\sqrt n}} = 1+o(1).
\]
Multiplying this by~\rf(eqn:one), we obtain~\rf(eqn:prodComparison).

\subsection{A Simple Lemma}
\label{sec:simple}
In this section, we prove a simple lemma that will be used in the proofs of both Lemmata~\ref{lem:one} and~\ref{lem:zero}.
Let $\gamma = \omega(\sqrt n)$.
Define a graph $G$ on the vertex set $X$ defined in \rf(sec:patternsproof), where two vertices $x$ and $y$ are connected iff $|x\cap y| \ge n/2 - \gamma$.

\begin{lem}
\label{lem:intersections}
For every non-empty connected subset $A$ of vertices of $G$, we have
\begin{equation}
\label{eqn:capA}
\absC|\bigcap_{x\in A} x| \ge \frac n2 - O\sA[|A|\gamma]
\qquad\text{and}\qquad
\absC|\bigcup_{x\in A} x| \le \frac n2 + O\sA[|A|\gamma].
\end{equation}
\end{lem}

\pfstart
We prove the first equality in~\rf(eqn:capA), the second one being similar.
The proof is by induction on the size of $A$.
The base case $|A|=1$ follows from the fact that $x\in X$ is nearly balanced.

For the inductive step, take a vertex $y\in A$ such that $A\setminus\{y\}$ is connected.  Let $z$ be a neighbour of $y$ in $A\setminus\{y\}$.
By the inductive hypothesis,
\[
\absC|\bigcap_{x\in A\setminus\{y\}} x| \ge \frac n2 - O\sA[(|A|-1) \gamma].
\]
Also,
\[
\absC|\bigcap_{x\in A\setminus\{y\}} x| - \absC|\bigcap_{x\in A} x| \le |z| - |z\cap y| = O(\gamma),
\]
since $z$ is nearly balanced and $|y\cap z|\ge n/2 - \gamma$.  Combining the last two inequalities, we obtain~\rf(eqn:capA).
\pfend

\subsection{Proof of \rf(lem:one)}
\label{sec:one}
\mycommand{compO}{\comp_{\mathbf{1}}}
\mycommand{compSO}{\compS_{\mathbf{1}}}
\mycommand{compTO}{\compT_{\mathbf{1}}}
As we only work with $J_1$ in this section, let, for $T\subseteq[n]$,
\[
\compTO = \prod_{j\in J_1} \compT_j.
\]
This is the projection of $\compT$ from \rf(sec:patternsproof) onto the indices in $J_1$.
For $j\in J_1$, let us denote
\[
y_j = \bigcap_{x \in X_{j,1}} x,
\qquad\text{and}\qquad
z_j = \bigcup_{x \in X_{j,1}} x.
\]
Using \rf(lem:intersections) with $\gamma = n^{3/4}$, we get
\[
\frac n3 \le \frac n2 - O(q n^{3/4}) \le |y_j| \le \frac n2 + O\sA[n^{3/4}]
\qquad\text{and}\qquad
|z_j| \le \frac n2 + O\sA[|X_{1,j}|n^{3/4}],
\]
if $n$ is large enough.
We impose the following constraints on $S\sim\cB$:
\begin{equation}
\label{eqn:SoneCondition}
|S| = O(\sqrt n),
\quad
|S\cap y_j|\le \frac{\sqrt n}2 + O\s[n^{1/4}\sqrt{\log n}]
\quad\text{and}\qquad
|S\setminus z_j|\ge \frac{\sqrt n}2 - O\s[|X_{1,j}|n^{1/4}\sqrt{\log n}]
\end{equation}
for all $j\in J_1$.
For $S\setminus z_j$, we have
\[
\bE\skA[|S\setminus z_j|] \ge \frac{\sqrt n}2 - O\s[|X_{1,j}|n^{1/4}]
\qquad
\mbox{and}
\qquad
\Var\skA[|S\setminus z_j|] \le \sqrt n.
\]
And similar estimates can be obtained for $S\cap y_j$.  Applying Bernstein's inequality and the union bound, we get that a set $S\sim\cB$ satisfies~\rf(eqn:SoneCondition) with probability $1-o(1)$ if the $O$ factors in~\rf(eqn:SoneCondition) are large enough.
In the remaining part of this section, we assume that $S$ satisfies~\rf(eqn:SoneCondition).

For each $(C_j)\in\compO$, we have $C_j(a) \in y_j$ for each $j\in J_1$ and $a\in[\sqrt n]$.
Also, for each $j\in J_1$ and $x\in X_{0,j}$, there exists $a\in[\sqrt n]$ such that $C_j(a) \notin x$.  
We call the smallest such $a$ the \emph{pivotal index} corresponding to $j$ and $x$.
We call $C_j(a)$ the corresponding \emph{pivotal element}.

Let $S_j$ be an arbitrary subset of $S\cap y_j$ of size
\begin{equation}
\label{eqn:Sjsize}
|S_j| = \max\sfig{0, |S\cap y_j| - |S\setminus z_j|} = O\s[|X_{1,j}|n^{1/4}\sqrt{\log n}].
\end{equation}
We define two auxiliary subsets of $\cC^{J_1}$.

\itemstart
\item A sequence $(C_j)\in \compO$ is called \emph{half-restricted} if all its pivotal elements lie outside of $S$.  Denote the set of half-restricted $(C_j)$ by $\cH^S$.
\item A sequence $(C_j)\in \cH^S$ is called \emph{restricted} if for all $j\in J_1$ and $a\in[\sqrt n]$ we have $C_j(a)\notin S_j$.
Denote the set of restricted $(C_j)$ by $\cR^S$.
\itemend

\rf(lem:one) follows from the following three claims.

\begin{clm}
We have $\abs|\compSO|\ge \abs|\cR^S|$.
\end{clm}

\pfstart
This is achieved by \emph{shifting}: moving elements from $S\cap y_j$ to $S\setminus z_j$.
More precisely, let $\pi_j\colon S\cap y_j\setminus S_j\to S\setminus z_j$ be any injective mapping.  It exists due to~\rf(eqn:Sjsize).  
Define a mapping $\pi\colon \cR^S\to\compSO$ as $\pi\colon (C_j)\mapsto (C'_j)$, where
\[
C'_j(a) = 
\begin{cases}
\pi_j(C_j(a)),& \text{if $C_j(a) \in S\cap y_j\setminus S_j$;}\\
C_j(a),&\text{if $C_j(a) \in y_j\setminus S$.} 
\end{cases}
\]
It is clearly an injective mapping.  Also, its image is a subset of $\compSO$ since, 
\itemstart
\item $C'_j(a) \in x^S$ for each $j\in J_1$, $a\in[\sqrt n]$ and $x\in X_{1,j}$; and
\item for each pivotal element $C_j(a)$ corresponding to $j\in J_1$ and $x\in X_{0,j}$, we have $C'_j(a) = C_j(a)\notin x^S$, ensuring $f_{C'_j}(x^S) = 0$.\qedhere
\itemend
\pfend

\begin{clm}
With probability $1-o(1)$ over the choice of $S\sim\cB$, we have $|\compO|\le \sA[1+o(1)]\abs|\cH^S|$.
\end{clm}

\pfstart
For each $i\in[n]$, let $d_i$ denote the number of sequences $(C_j)\in\compO$ for which $i$ is a pivotal element.  Since each $(C_j)$ has at most $q^2$ pivotal elements, we see that
\[
\sum_{i\in[n]} d_i \le |\compO|q^2 = o\s[|\compO|\sqrt{n}].
\]
In particular,
$
\bE_{S\sim\cB} \skB[\sum_{i\in S} d_i] = o(|\compO|).
$
By Markov's inequality, with probability $1-o(1)$, we have $\sum_{i\in S} d_i = o(|\compO|)$, implying the claim.
\pfend

\begin{clm}
We have $|\cH^S|\le \sA[1+o(1)]|\cR^S|$.
\end{clm}

\pfstart
In this case, it is easier to consider each $j\in J_1$ independently.
Again, the conditions for different $j$ are independent, hence,
\[
|\cR^S| = \prod_{j\in J_1} |\cR^S_j|,
\qquad\text{and}\qquad
|\cH^S| = \prod_{j\in J_1} |\cH^S_j|,
\]
where $\cR^S_j$ and $\cH^S_j$ are the projections of $\cR^S$ and $\cH^S$ onto the $j$th component.
We prove that
\begin{equation}
\label{eqn:RH1}
|\cH^S_j| \le \ee^{O\s[n^{-1/4}\sqrt{\log n}|X_{1,j}|]} |\cR^S_j|,
\end{equation}
which implies the claim, since then
\[
\frac{|\cH^S|}{|\cR^S|} 
\le \ee^{O\s[n^{-1/4}\sqrt{\log n}\sum_{j\in J_1} |X_{1,j}|]} 
\le \ee^{O\s[n^{-1/4}\sqrt{\log n}\cdot q]} = 1+o(1).
\]

Let $\cH^S_{j,k}$ denote the subset of $C'\in \cH^S_{j}$ such that for exactly $k$ values of $a\in[\sqrt n]$, we have $C'(a) \in S_j$.  In particular, $\cR^S_{j} = \cH^S_{j,0}$, and $\cH^S_{j} = \bigcup_k \cH^S_{j,k}$.
We say that a clause $C\in \cR^S_{j}$ is \emph{in relation} with a clause $C'\in\cH^S_{j,k}$ iff $C(a) = C'(a)$ whenever $C'(a)\notin S_j$.

Clearly, each clause $C\in \cR^S_{j}$ is in relation with at most $\binom{\sqrt n}k |S_j|^k$ clauses in $\cH^S_{j,k}$.
Next, we claim that if we take a clause $C'\in \cH^S_{j,k}$ and substitute each $C'(a)\in S_j$ with an element of $y_k\setminus S$, we get a clause $C\in \cR^S_j$.
First, any $C'(a)$ that was changed was not a pivotal element of $C'$ since $S_j\subseteq S$ and $C'\in\cH^S_j$.
Next, a new element $C(a)$ can become a pivotal element of $C$, but it lies outside of $S$, so, nonetheless, $C\in\cR^S_j$.
Hence, each clause $C'\in\cH^S_{j,k}$ is in relation with at least $\abs|y_j\setminus S|^k$ clauses in $\cR^S_j$.
Using double counting,
\[
\frac{|\cH^S_{j,k}|}{|\cR^S_{j}|}
\le\frac{\binom{\sqrt n}k |S_j|^k}{|y_j\setminus S|^k}
\le\frac{n^{k/2}/k!\cdot \sA[O(|X_{1,j}|n^{1/4}\sqrt{\log n})]^k}{\sA[\Omega(n)]^k}
= \frac1{k!} \s[{O\sB[\frac{|X_{1,j}|\sqrt{\log n}}{n^{1/4}}]}]^k.
\]
Hence,
\[
\frac{|\cH^S_{j}|}{|\cR^S_{j}|} 
\le 1+\sum_{k\ge 1} \frac1{k!} \s[{O\sB[\frac{|X_{1,j}|\sqrt{\log n}}{n^{1/4}}]}]^k
=\ee^{O\s[n^{-1/4}\sqrt{\log n}|X_{1,j}|]}. \qedhere
\]
\pfend

\subsection{Proof of \rf(lem:zero)}
\label{sec:zero}
As in \rf(sec:patternsproof), we treat pairs of inputs that are far from each other separately.
Let a parameter $\gamma = \Theta(\sqrt{n}\log n)$ be specified later.
Define the graph $G$ as in \rf(sec:simple).
Let $G_1,\dots,G_\varkappa$ be the connected components of $G$, and
\[
z_k = \bigcap_{x\in G_k} x.
\]
Using \rf(lem:intersections), we get that
\[
|z_k| \ge \frac n2 - O(q \gamma)\ge \frac n2 - O(n^{3/4}).
\]
We impose the following constraints on $S\sim\cB$:
\begin{equation}
\label{eqn:SzeroCondition}
|S| = O(\sqrt n),
\quad
|S\cap x| \ge \frac{\sqrt{n}}2 - O\s[n^{1/4}\sqrt{\log n}],
\quad\text{and}\quad
|S\setminus z_k| \le \frac{\sqrt{n}}2 + O\s[n^{1/4}\sqrt{\log n}]
\end{equation}
for all $x\in X$ and $k\in[\varkappa]$.
By Bernstein's inequality again, if the $O$ factors are chosen appropriately, a set $S\sim\cB$ satisfies~\rf(eqn:SzeroCondition) with probability $1-o(1)$.

\mycommand{posit}{\cP^\emptyset}
\mycommand{positS}{\cP^S}
\mycommand{positT}{\cP^T}
Let us assume up to the end of the section that a subset $S$ satisfying~\rf(eqn:SzeroCondition) is fixed.
We say a clause $C$ is \emph{positive with respect to the shift} $T\in\{\emptyset, S\}$ iff $f_C(x^T)=1$ for some $x\in X'$.
Denote the set of such clauses by $\positT$.

In this section, we are also going to treat a clause $C\colon[\sqrt n]\to[n]$ as a multiset.
For example, if $A\subseteq[n]$, we denote by $C\cap A$ the partial function from $[\sqrt n]$ to $[n]$, defined by $(C\cap A)(a) = C(a)$ for all $a\in[\sqrt n]$ such that $C(a)\in A$, and not defined for the remaining $a$.  
We call such functions \emph{partial clauses}.  
A partial clause $C\setminus A$ is defined similarly.
The \emph{size} of a partial clause is the size of its domain.  We say that a partial clause is contained in $A\subseteq[n]$ if its range is contained in $A$.

\begin{clm}
\label{clm:CcapS}
We have
\[
\Pr_{C\sim\positS} \skB[|C\cap S| \ge \Omega(\log n)] \le \frac1n.
\]
\end{clm}

\pfstart
This holds because $|C\cap S|$ approximately follows a Poison distribution.
Indeed, for a non-negative integer $k$, let 
\[
\positS_k = \sfig{C\in\positS \midA |C\cap S| = k}.
\]
We say that $C\in\positS_0$ is \emph{in relation} with $C'\in\positS_k$ iff $C(a) = C'(a)$ for all $a$ such that $C'(a)\notin S$.

Each $C\in\positS_0$ is in relation with at most $\binom{\sqrt n}{k} |S|^k$ clauses in $\positS_k$.  
On the other hand, let $C'\in\positS_k$.
Then, there exists $x\in X'$ such that $C'\subseteq x^S$.
Hence, $C'$ is in relation with at least $|x\setminus S|^k = \sA[\Omega(n)]^k$ clauses in $\positS_0$.
Using double counting,
\[
\frac{|\positS_k|}{|\positS_0|} \le \frac{\binom{\sqrt n}{k} \sA[O(\sqrt n)]^k}{\sA[\Omega(n)]^k} = \frac{\sA[O(1)]^k}{k!}.
\]
This implies the claim.
\pfend

Thus, we can only focus on those $C$ that have small intersection with $S$.
Let $B$ be a partial clause with $B\subseteq[n]\setminus S$ and $\sqrt{n}-O(\log n)\le |B|\le \sqrt{n}$. For $T\in\{\emptyset,S\}$, let us denote
\[
\positT_B = \{C\in \positT \mid C\setminus S = B\}.
\]

We say $B$ is {\em bad} if $B\subseteq x\cap y$ where $x, y\in X'$ are vertices from different connected components of $G$.  Otherwise, we call $B$ \emph{good}.

\begin{lem}
\label{lem:goodB}
If $B$ is good, then
\[
|\posit_B| \ge \s[1- O{\sA[n^{-1/4}\log^{3/2}n]}] |\positS_B|.
\]
\end{lem}

\pfstart
Let $X_B = \{x\in X'\mid B\subseteq x\}$.
If $X_B$ is empty, then both $\posit_B$ and $\positS_B$ are empty, and we are done, so assume there is some $x\in X_B$.
Let $y_B = \bigcap_{y\in X_B} y$.
Also, as $B$ is good, $X_B$ is contained in some connected component $G_k$ of $G$.
In particular, $z_k\subseteq y_B$.

\mycommand{DComp}{\bar{D}}
Let $\DComp\subseteq[\sqrt n]$ be the complement of the domain of $B$.
In particular, $|\DComp| = O(\log n)$.
The size of $\posit_B$ is at least the number of functions from $\DComp$ to $x\cap S$,
and the size of $\positS_B$ is at most the number of functions from $\DComp$ to $S\setminus y_B \subseteq S\setminus z_k$.
Thus, using~\rf(eqn:SzeroCondition):
\[
\frac{|\posit_B|}{|\positS_B|} \ge \s[\frac{\frac{\sqrt n}2 - O(n^{1/4}\sqrt{\log n})}{\frac{\sqrt n}2 + O(n^{1/4}\sqrt{\log n})}]^{O(\log n)} \ge 1- O(n^{-1/4}\log^{3/2}n).\qedhere
\]
\pfend

\begin{lem}
\label{lem:badB}
We have
\[
\absC| \bigcup_{\text{\rm $B$ is bad}} \positS_B | \le O\sB[\frac1n]\, |\positS|.
\]
\end{lem}

\pfstart
Fix a particular pair $x,y\in X'$ of vertices that lie in different connected components of $G$.  
Then, for each $\positS_B$ with $B\subseteq x\cap y$, and each clause $C\in \positS_B$, we have $C \subseteq (x\cap y) \cup S$.  
On the other hand, we may lower bound the number of clauses contained in $\positS$ by the number of clauses contained in $x^S$. 
Thus,
\begin{equation}
\label{eqn:badB}
\frac{\abs|\bigcup_{B\subseteq x\cap y} \positS_B |} {|\positS|}
\le \s[ \frac{\absA| (x\cap y) \cup S |}{|x^S|} ]^{\sqrt n}
\le \s[ \frac{\frac n2 - \Omega(\gamma)}{\frac n2 - O(\sqrt{n})} ]^{\sqrt n}
\le \sC[1 - \Omega{\sB[\frac\gamma n]}]^{\sqrt n} \le \ee^{-\Omega(\gamma/{\sqrt n})}.
\end{equation}
Taking the $\Theta$-factor in the definition of $\gamma$ sufficienly large and summing~\rf(eqn:badB) over all $x$ and $y$, we obtain the lemma.
\pfend

Thus, using Lemmata~\ref{lem:goodB}, \ref{lem:badB} and \rf(clm:CcapS), 
\[
|\posit| \ge \sum_{\text{$B$ is good}} |\posit_B|
\ge \sA[1 - O(n^{-1/4}\log^{3/2}n)] \sum_{\text{$B$ is good}} |\positS_B|
\ge \sA[1 - O(n^{-1/4}\log^{3/2}n)] |\positS|,
\]
proving \rf(lem:zero).

\section{Testing Monotonicity of Regular LTFs}
\label{sec:upper}

\subsection{Randomized Bisection Process}

The key component of the analysis of the bisection algorithm and the proof
of Theorem~\ref{thm:regular-LTFs} is the analysis
of \emph{randomized bisection processes}, as defined below.

\begin{definition}[Randomized bisection process]
	Fix any finite set $S$.
	The \emph{randomized bisection process} with initial set $S$
	is the sequence of random sets $S_0,S_1,S_2,\ldots$
	defined as follows. Initially, $S_0 = S$.
	For each $k \ge 1$, $S_{k-1}$ is partitioned uniformly at random into two sets
	$A_k$ and $B_k$. Then the set $S_k$ is chosen to be either $A_k$ or $B_k$ by some
	arbitrary (and possibly adversarial) external process.
\end{definition}


\begin{lemma}
\label{lem:bisection-process}
	For any $\delta>0$, there exists $\kappa = \kappa(\delta)$ such that 
	with probability at least $1-\delta$,
	the randomized bisection process $S_0,S_1,S_2,\ldots$ with initial set $S$ satisfies 
	\begin{equation}
	\label{eqn:bisection}
	\frac12\cdot \frac{|S|}{2^k} < |S_k| < \frac32 \cdot \frac{|S|}{2^k}
	\end{equation}
	for every $k \le \log |S| - \kappa$.
\end{lemma}

\begin{proof}
Let us prove the lower bound first.
It is clear that the best strategy for the adversary is to take the smallest of $A_k$ and $B_k$ on each step, so we may assume that the sets $S_k$ of size less than $|S_{k-1}|/2$ have double probability to appear, whereas the sets $S_k$ of size more than $|S_{k-1}|/2$ never appear at all.

Using Fubini's theorem and the Chernoff-Hoeffding bound, we obtain, for a fixed $S_{k-1}$, 
\begin{align*}
\E\skA[|S_{k}|] 
&= \frac{|S_{k-1}|}2 - 2 \int_{0}^{+\infty} \Pr\skB[\cB < \frac{|S_{k-1}|}2 - t]\, dt \\
&\ge \frac{|S_{k-1}|}2 - 2 \int_{0}^{+\infty} \ee^{-2{t^2}/{|S_{k-1}|}} \,dt
\ge \frac{|S_{k-1}|}2 - O\s[\sqrt{|S_{k-1}|}],
\end{align*}
where $\cB$ is the binomial probability distribution on $|S_{k-1}|$ elements with probability $\frac12$.

Since $x\mapsto \frac x2 - c\sqrt{x}$ is a convex function, if we unfix $S_{k-1}$, we get by Jensen's inequality
\begin{equation}
\label{eqn:bis1}
\E\skA[|S_{k}|] \ge \frac{\E\skA[|S_{k-1}|]}2 - O\s[\sqrt{\E\skA[|S_{k-1}|]}].
\end{equation}
It is clear that $|S_k|\le |S|/2^k$, thus, by induction on $k$,
\begin{equation}
\label{eqn:bis2}
\E\skA[|S_{k}|] 
\ge \frac{|S|}{2^k} - \sum_{j=1}^k \frac1{2^{k-j}}\cdot O\s[\sqrt{\frac{|S|}{2^{j-1}}}]
= \frac{|S|}{2^k} - \sqrt{\frac{|S|}{2^k}} \sum_{j=1}^k \frac{O(1)}{2^{(k-j)/2}}
\ge \frac{|S|}{2^k} - O\s[\sqrt{\frac{|S|}{2^k}}].
\end{equation}
If $\frac{|S|}{2^k} = \Omega_\delta(1)$, we have
\[
\E\skA[|S_{k}|] 
> \frac{|S|}{2^k} - \frac{\delta}{4}\cdot \frac{|S|}{2^k}.
\]
And since $|S_k| \le |S|/2^k$, we have by Markov's inequality that
\[
\Pr\skB[|S_k| \le \frac12\cdot \frac{|S|}{2^k}] \le \frac\delta2.
\]

The proof of the upper bound is similar. This time the adversary takes the largest of $A_k$ and $B_k$.
Similarly to~\rf(eqn:bis1), we get
\[
\E\skA[|S_{k}|] \le \frac{\E\skA[|S_{k-1}|]}2 + O\s[\sqrt{\E\skA[|S_{k-1}|]}].
\]
We show by induction on $k$ that if $\frac{|S|}{2^k} = \Omega(1)$, then $\E[|S_k|] \le \frac32\cdot\frac{|S|}{2^k}$.  This is done similarly to~\rf(eqn:bis2):
\[
\E\skA[|S_{k}|] 
\le \frac{|S|}{2^k} + \sum_{j=1}^k \frac1{2^{k-j}}\cdot O\s[\sqrt{\frac{3|S|}{2^{j}}}]
= \frac{|S|}{2^k} + \sqrt{\frac{|S|}{2^k}} \sum_{j=1}^k \frac{O(1)}{2^{(k-j)/2}}
\le \frac{|S|}{2^k} + O\s[\sqrt{\frac{|S|}{2^k}}].
\]
Again, if $\frac{|S|}{2^k} = \Omega_\delta(1)$, we have
\(
\E\skA[|S_{k}|] 
< \frac{|S|}{2^k} + \frac{\delta}{4}\cdot \frac{|S|}{2^k},
\)
and, since $|S_k|\ge {|S|}/{2^k}$,
\[
\Pr\skB[|S_k| \ge \frac32\cdot \frac{|S|}{2^k}] \le \frac\delta2.\qedhere
\]
\pfend

\subsection{Non-Monotonicity of LTFs}

\begin{proposition}
\label{prop:non-monotone-LTF}
	If $f\colon \{-1,1\}^n \to \{-1,1\}$ is a non-constant LTF with weights 
	$w_1,\ldots,w_n$ such that $\sum_{i : w_i < 0} |w_i| > \max_i w_i$, then
	$f$ is not monotone.
\end{proposition}

\begin{proof}
	Let $N = \{i \in [n] \mid w_i < 0\}$ denote the set of indices with negative 
	weights and let $\eta = \sum_{i \in N} |w_i|$.
	Let $X \in \{-1,1\}^n$ be the subset of inputs such that for every $i \in N$, $x_i = 1$.
	
	There exists $x \in X$ such that
	$\theta-2\eta \le \sum_{i \in [n]} w_i x_i < \theta$.
	Indeed, there exists an input 
	$x' \in X$ with $\sum_{i \in [n]} w_i x'_i < \theta$ (otherwise $f$ is the 
	constant 1 function), and an input
	$x'' \in X$ with $\sum_{i \in [n]} w_i x''_i \ge \theta-2\eta$
	(otherwise $f$ is the constant $-1$ function).
	Also, $\max_i |w_i| \le \eta$, hence, changing the value of one variable changes the value of the sum $\sum_{i \in [n]} w_i x_i$ by at most $2\eta$.
	
	With this choice of $x$, let $y \in \{-1,1\}^n$ be defined by $y_i = x_i$ for $i \in [n] \setminus N$
	and $y_j = -1$ for every $j \in N$. Then $\sum_{i \in [n]} w_i y_i \ge \theta$
	so $x \succeq y$ and $1 = f(y) > f(x) = -1$, hence, $f$ is non-monotone.
\end{proof}

In the proof of Theorem~\ref{thm:regular-LTFs}, we need to show that
regular LTFs that are far from monotone must have a large number of 
reasonably large negative weights. 
Using this lemma, we obtain the following bound on the magnitude of the negative
weights of regular LTFs that are far from monotone.

\begin{proposition}
\label{prop:regular-far-monotone}
  	Fix $n\ge 1$ and $\epsilon > 0$.
  	Let $f\colon \{-1,1\}^n \to \{-1,1\}$ be a $\tau$-regular LTF with weights $w_1,\ldots,w_n$ that is $\epsilon$-far from monotone.
		Assume $\sum_{i \in [n]} w_i^{\,2} = 1$ and $\tau \le \frac\eps4$.
  	Then the set
  	$N = \left\{i \in [n]\mid w_i < 0 \right\}$ 
  	of indices corresponding to negative weights satisfies 
  	$\sum_{i \in N} w_i^{\,2} \ge \frac{\epsilon^2}{256 \ln (8/\epsilon)}$.
\end{proposition}

\begin{proof}
Let $g\colon \{-1,1\}^n \to \{-1,1\}$ be the LTF 
$g(x) = \sgn(\sum_{i \in [n] \setminus N} w_i x_i - \theta)$ obtained by 
removing the negative weights of $f$. 
Since the function $g$ is monotone,
$$
\Pr[ f(x) \neq g(x) ] \ge \epsilon.
$$
The event $f(x) \neq g(x)$ can only occur when 
$|\sum_{i \in N} w_i x_i| > \left|\sum_{i \in [n] \setminus N} w_i x_i - \theta\right|$.
So for any $t > 0$,
\begin{align*}
\Pr[ f(x) \neq g(x) ] 
&\le
\Pr\skC[ \absB|\sum_{i \in [n] \setminus N} w_i x_i - \theta| \le t ]
+ \Pr\skC[ \absB|\sum_{i \in N} w_i x_i| > t ].
\end{align*}
Define $\eta = \sum_{i \in N} w_i^{\,2}$. If $\eta > \frac12$, then we are done.
So assume from now on that $\eta \le \frac12$.
Fix $t = \sqrt{ 2\eta \ln(\tfrac8\epsilon)}$.
By Lemma~\ref{lem:berry-esseen-cor}, 
$$
\Pr\skC[ \absB|\sum_{i \in [n] \setminus N} w_i x_i - \theta| \le \sqrt{ 2\eta\ln(\tfrac8\epsilon)}]
\le 2\sqrt{\frac{2\eta\ln(\tfrac8\epsilon)}{1-\eta}} + \frac{\epsilon}{2}
\le 4\sqrt{\eta\ln(\tfrac8\epsilon)} + \frac{\epsilon}{2}
$$
and by the Hoeffding bound,
$$
\Pr\skC[ \absB|\sum_{i \in N} w_i x_i| > \sqrt{ 2\eta\ln(\tfrac8\epsilon)} ]
\le 2\ee^{-\big(2\eta\ln(\tfrac8\epsilon)\big)/2\eta} \le \frac\epsilon4.
$$
Putting all the inequalities together, we obtain the inequality
$\frac\epsilon4 \le 4\sqrt{\eta\ln(\tfrac8\epsilon)}$, which is satisfied
if and only if $\eta \ge \frac{\epsilon^2}{256 \ln(\frac8\epsilon)}$.
\end{proof}

\begin{corollary}
\label{cor:regular-far-monotone}
  	Fix $\epsilon > 0$ and $\tau > 0$. 
  	There exists $n_0 = n_0(\epsilon,\tau)$ such that for every $n \ge n_0$,
  	if $f \colon \{-1,1\}^n \to \{-1,1\}$ is a $\frac\tau{\sqrt{n}}$-regular LTF with normalized weights $w_1,\ldots,w_n$ ($\sum_{i\in [n]} w_i^{\,2} = 1$)
  	that is $\epsilon$-far from monotone, then the set
  	$N^\dagger = \sfigB{i \in [n] \midB w_i < -\frac{\epsilon}{\sqrt{512\ln(\frac8\epsilon) n}} }$ 
  	has cardinality $|N^\dagger| \ge \frac{\epsilon^2}{512 \tau^2 \ln(\frac8\epsilon)} \cdot n$.
\end{corollary}

\begin{proof}
	Let $n_0$ be the minimal positive integer such that 
	$\frac\tau{\sqrt{n_0}} < \frac\epsilon4$.
	Define $N = \{i \in [n] \mid w_i < 0\}$.
	By Proposition~\ref{prop:regular-far-monotone}, 
	the sum of the squares of the negative weights is bounded below by
	$\sum_{i \in N} w_i^{\,2} \ge \frac{\epsilon^2}{256 \ln(\frac8\epsilon)}$.
	For every element $i$ in $N \setminus N^\dagger$, the weight $w_i$ satisfies
	$w_i^{\,2} \le \frac{\epsilon^2}{512 \ln(\frac8\epsilon)n}$ so
	$
	\sum_{i \in N \setminus N^\dagger} w_i^{\,2} \le |N \setminus N^\dagger| \cdot \frac{\epsilon^2}{512 \ln(\frac8\epsilon)n} \le \frac{\epsilon^2}{512 \ln(\frac8\epsilon)}
	$
	and
	$$
	\sum_{i \in N^\dagger} w_i^{\,2} =
	\sum_{i \in N} w_i^{\,2} - \sum_{i \in N \setminus N^\dagger} w_i^{\,2} \ge 
	\frac{\epsilon^2}{256 \ln(\frac8\epsilon)} - \frac{\epsilon^2}{512 \ln(\frac8\epsilon)} = \frac{\epsilon^2}{512 \ln(\frac8\epsilon)}.
	$$
	The regularity of $f$ guarantees that
	$
	\sum_{i \in N^\dagger} w_i^{\,2}
	\le |N^\dagger| \frac{\tau^2}{n}
	$
	and so
	$
	|N^\dagger| \ge \frac{n}{\tau^2} \cdot \frac{\epsilon^2}{512 \ln(\frac8\epsilon)}.
	$
\end{proof}

\subsection{Proof of Theorem~\ref{thm:regular-LTFs}}
\label{sec:algProof}

	Let $f$ be any $\frac{\tau}{\sqrt{n}}$-regular LTF that is 
	$\epsilon$-far from monotone.
	We may assume that $\sum_i w_i^{\,2} = 1$.
  We modify the bisection algorithm from \rf(alg:bisection) to make the analysis easier.
	Define 
	\[
	c = \frac{\epsilon^2}{512 \tau^2 \ln(\frac8\epsilon)}
\qquad\text{and}\qquad
  \zeta = \frac{\epsilon}{\sqrt{512\ln(\frac8\epsilon)}}.
	\]
	
Let
\[
 k = \floorB[ \log(cn) - \max\sfigB{\log(\tfrac {8\tau}\zeta), \kappa(\tfrac18) } ],
\]
where $\kappa$ is as in \rf(lem:bisection-process).

Consider \rf(alg:modified).
Clearly, the algorithm never rejects a monotone function.
Let us now assume that $f$ is $\eps$-far from monotone.
By Corollary~\ref{cor:regular-far-monotone}, the set $N = \{i \in [n] \mid w_i < -\frac{\zeta}{\sqrt{n}}\}$ then has cardinality $|N| \ge cn$.

\begin{algorithm}
	\caption{Modified Bisection Algorithm}
	\label{alg:modified}
	\begin{algorithmic}[1]
		\State Draw $x \in \sbool^n$ uniformly at random.
		\State Draw $y \in \sbool^n$ uniformly at random $8/\eps$ times or until $f(x)\ne f(y)$.
		\State If no $y$ satisfying the condition $f(x)\ne f(y)$ was found, {\bf accept}.
		\State Assume $f(x)=-1$ and $f(y)=1$.  Otherwise, swap $x$ and $y$.
		\For{$k$ times}
			\State Draw $z \in \Hybrid(x,y)$ uniformly at random.
			\State If $f(z) = -1$, update $x \gets z$.
			\State Otherwise if $f(z) = 1$, update $y \gets z$.
		\EndFor
		\State If $\abs| x\cap y | > \frac32\cdot \frac{n}{2^k}$, {\bf accept}.
		\State Query all inputs in $\Hybrid(x,y)$.  {\bf Reject} if a non-monotone edge found; otherwise {\bf accept}.
	\end{algorithmic}
\end{algorithm}

Assume $x\in \sbool^n$ is fixed, and $y$ is uniformly sampled from $\sbool^n$.
First, $f$ is $\eps$-far from a constant function, hence, an $y$ satisfying $f(x)\ne f(y)$ will be found with probability at least $\frac78$.
Next, let $x\bigtriangleup y$ be the set of indices where $x$ and $y$ differ.
By Chernoff bound, the probability that $\abs|(x\bigtriangleup y)\cap N|< cn/4$ is $o(1)$.
Thus, with probability at least $\frac 34$, after Step~4, $x$ and $y$ satisfy $f(x)=-1$, $f(y)=1$ and the set $S= x\bigtriangleup y$ satisfies $|S\cap N|\ge cn/4$.

Let $x_\ell$ and $y_\ell$ denote the value of $x$ and $y$ after the $\ell$th iteration of the loop in \rf(alg:modified).  In particular, $x_0$ and $y_0$ are the inputs $x$ and $y$ after Step~4 as in the previous paragraph.
	Denote $S_\ell = x_\ell\bigtriangleup y_\ell$ and $N_\ell = N \cap S_\ell$.
	The sets $S_0,S_1,S_2,\ldots$ and the sets $N_0,N_1,N_2,\ldots$ are 
	randomized bisection processes with the initial sets $S$ and $N\cap S$, respectively.
	By Lemma~\ref{lem:bisection-process}, with probability
	at least $\frac14$, the sets $S_k$ and $N_k$ satisfy
	$$
	|S_k| < \frac32 \cdot \frac{|S|}{2^k}
	\le O\sB[\frac{n}{cn \frac\zeta\tau}]
	= O\sB[\frac\tau{c\zeta}]
	$$
	and
	$$
	|N_k| > \frac12 \cdot \frac{|N\cap S|}{2^k} 
	\ge \frac12 \cdot \frac{cn/4}{cn \frac\zeta{8\tau}}
	\ge \frac\tau{\zeta}.
	$$
	In turn, this implies that the sum of the weights with coordinates in $N_k$ satisfies 
	$$
	\sum_{i \in N_k} |w_i| \ge |N_k| \cdot \min_{i\in N_k} |w_i| > \frac\tau\zeta\cdot\frac\zeta{\sqrt n} = \frac\tau{\sqrt n}\ge \max_j |w_j|.
	$$ 
	Let $f_{x,y}$ denote the function $f$ restricted to the set $\Hybrid(x,y)$, where $x$ and $y$ are as in Step~10 of the algorithm.
	This function is non-constant since since $f_{x,y}(x) = -1$	and $f_{x,y}(y) = 1$.
	By Proposition~\ref{prop:non-monotone-LTF}, 
	$f_{x,y}$ is a non-monotone LTF on $|S_k| = O\sA[\frac\tau{c\zeta}]$
	variables. 
	Then the algorithm rejects in Step~11 after additional $2^{|S_k|} = 2^{\tO(\tau^3/\eps^3)}$ queries.

\subsection{Truncated Functions}
\label{sec:algBalance}
In \rf(sec:algProof), we showed that the bisection algorithm efficiently $\eps$-tests regular LTFs for monotonicity, as specified by \rf(thm:regular-LTFs).
However, in the actual lower bounds by Fischer \etal~\cite{fischer:monotonicitytesting}, and Chen \etal~\cite{chen:newAlgorithmsLowerBoundsMonotonicity, chen:booleanMonotonicityRequiresSqrtn}, \emph{truncated} LTFs are used, as in \rf(defn:truncate).
In this section, we show that if the bisection algorithm can test some class of functions for monotonicity, then it can also test the truncated version of the same class with a modest slow-down.

Towards this goal, we argue that with probability $\Omega_\eps(1)$, the bisection algorithm only queries inputs in the middle layers of the cube, i.e., satisfying $\frac n2 - \delta\sqrt n\le |x|\le \frac n2 + \delta\sqrt n$ in the notation of \rf(defn:truncate).  
It is easy to modify the parameters of \rf(alg:modified) in \rf(sec:algProof) so that the algorithm uses $O_{\eps,\tau,p}(1)+\log n$ queries, always accepts a monotone function, and rejects a non-monotone $\frac\tau{\sqrt n}$-regular LTF with probability at least $1-p$.  
Combining the two statements, we get an algorithm that $\eps$-tests truncated $\frac\tau{\sqrt n}$-regular LTFs for monotonicity in $O_{\eps,\tau}(\log n)$ queries.

Consider the performance of \rf(alg:modified) on a function of the form $\truncate_\delta(f)$.  
The algorithm does not know the value of $\delta$, but it knows $\eps$, the distance from a non-monotone function $\truncate_\delta(f)$ to the closest monotone function.  
By \rf(lem:berry-esseen-cor), there exists $\beta = \beta(\eps) > 0$ such that with probability at least $\frac\eps2$, the input $y$ found on Step 2 of the algorithm satisfies
\[
\frac n2 - (\delta-\beta) \sqrt n\le |y|\le \frac n2 + (\delta-\beta)\sqrt n.
\]
The input $x$ also satisfies the same estimates with probability $\Omega_\eps(1)$.

Let $x_\ell$, $y_\ell$ and~$S_\ell$ be as in \rf(sec:algProof).  Let also $z_\ell$ denote the input $z$ on the $(\ell+1)$st iteration of the loop, so that either $x_{\ell+1}$ or $y_{\ell+1}$ equals $z_\ell$.
We consider those executions of the algorithm, in which
\begin{equation}
\label{eqn:trunc}
\max\sfig{
\absC|{\absA|z_\ell \cap S_\ell\cap x_\ell | - \frac{|S_\ell\cap x_\ell|}2}|,\;\;
\absC|{\absA|z_\ell \cap S_\ell\setminus x_\ell | - \frac{|S_\ell\setminus x_\ell|}2}|}
\le \frac{\beta(1-\alpha)}4 \alpha^\ell \sqrt n 
\end{equation}
for all $\ell\le k$, where $k=\ceil[\log (4\sqrt n/{\beta})]$ and $\alpha = 0.9$ (or any other constant $\frac1{\sqrt 2}<\alpha<1$).

We first show that~\rf(eqn:trunc) is satisfied for all $\ell\le k$ with probability $\Omega_\eps(1)$.
By induction, for each $\ell$,
\begin{equation}
\label{eqn:trSell}
|S_\ell| \le \frac n{2^\ell} + \frac{\beta(1-\alpha)}2 \sum_{i=0}^{\ell-1} \frac{\alpha^i}{2^{\ell-i-1}} \sqrt n
\le \frac n{2^\ell} + \frac{\beta(1-\alpha)}{2\alpha-1} \alpha^\ell \sqrt n \le 2\cdot \frac n{2^{\ell}},
\end{equation}
where the last inequality holds if $n$ is large enough.
By the Chernoff-Hoeffding bound, the probability that~\rf(eqn:trunc) is satisfied for all $\ell\le k$ is at least
\[
\prod_{\ell=0}^{k} \s[{1 - 2\exp\sB[-2\frac{\frac{\beta^2(1-\alpha)^2}{16}\alpha^{2\ell}n }{2\cdot \frac n{2^\ell}}] }]^2
 \ge \s[ \prod_{\ell=0}^\infty {\sB[1-2\ee^{-\Omega_\eps\sA[(2\alpha^2)^\ell]}]} ]^2,
\]
and the infinite product converges.

By induction again, for each $\ell\le k$,
\[
\abs|{|z_\ell| - \frac n2}| < (\delta-\beta)\sqrt n + \frac{\beta(1-\alpha)}2 \sum_{\ell=0}^{+\infty} \alpha^\ell \sqrt n = \sA[\delta - \frac\beta 2]\sqrt n.
\]
Also, by~\rf(eqn:trSell), $|S_k|\le \frac{\beta}2\sqrt n$. Hence, all the inputs queried by the algorithm after the $k$th iteration of the loop are also in the middle layers of the cube.

\subsection*{Acknowledgements}
A.B. is supported by FP7 FET Proactive project QALGO.
Part of this work was completed while A.B. was at the University of Latvia.
E.B. is supported by an NSERC Discovery grant.

\bibliographystyle{habbrvM}

\bibliography{belov}

\end{document}